\documentclass[journal]{IEEEtran}

\usepackage{setspace}

\usepackage{graphics,
	psfrag,
	epsfig,
	amsthm,
	cite,
	amssymb,
	url,
	dsfont,
	algorithm,
	algorithmic,
	balance,
	enumerate,
	color,
	setspace
}
\usepackage{amsmath}

\usepackage{epstopdf}

\newtheorem{theorem}{Theorem}
\newtheorem{corollary}{Corollary}
\newtheorem{proposition}{Proposition}

\newcommand{\eref}[1]{(\ref{#1})}
\newcommand{\sref}[1]{Section~\ref{#1}}

\newcommand{\pref}[1]{Proposition~\ref{#1}}
\newcommand{\cref}[1]{Constraint~\ref{#1}}
\newcommand{\corref}[1]{Corollary~\ref{#1}}
\newcommand{\thref}[1]{Theorem~\ref{#1}}

\newcommand{\tref}[1]{Table~\ref{#1}}
\newcommand{\algref}[1]{Algorithm~\ref{#1}}

\newcommand{\ignore}[1]{}

\IEEEoverridecommandlockouts
\usepackage{subcaption}
\begin{document}


\title{\vspace{-.36cm} Cross-Layer Network Codes for Content Delivery in Cache-Enabled D2D Networks}

\author{Mohammed S. Al-Abiad, Student Member, IEEE, Md. Zoheb Hassan, Student Member, IEEE, and Md. Jahangir Hossain, Senior Member, IEEE
	
	
	\thanks{Mohammed S. Al-Abiad, Md. Zoheb Hassan, and Md. Jahangir Hossain are with the School of
		Engineering, University of British Columbia, Kelowna, BC V1V 1V7, Canada
		(e-mail: m.saif@alumni.ubc.ca, zohassan@mail.ubc.ca, jahangir.hossain@ubc.ca).
		
		}
	\vspace{-0.8cm}}
\maketitle

\IEEEoverridecommandlockouts

\begin{abstract}

In this paper, we consider the use of cross-layer network coding (CLNC), caching,  and  device-to-device (D2D) communications to jointly optimize the delivery of a set of popular contents to a set of user devices (UDs). In the considered D2D network, a group of near-by UDs cooperate with each other and use NC to combine their cached files, so as the completion time required for delivering all requested contents to all UDs is minimized. Unlike the previous work that considers only one transmitting UD at a time, our work allows multiple UDs to transmit simultaneously given the interference among the active links is small. Such configuration brings a new trade-off among scheduling UDs to transmitting UDs, selecting the coding decisions and the transmission rate/power. Therefore, we consider the completion time minimization problem that involves scheduling multiple transmitting UDs, determining their transmission rates/powers and file combinations. The problem is shown to be intractable because it involves all future coding decisions. To tackle the problem at each transmission slot, we first design a graph called herein
the D2D Rate-Aware IDNC graph where its vertices have weights that judiciously balance between the rates/powers of the transmitting UDs and the number of their scheduled UDs. Then, we propose an innovative and efficient CLNC solution that iteratively selects a set of transmitting UDs only if the interference caused by the transmissions of the newly selected UDs does not significantly impact the overall completion time.  Simulation results show that the proposed solution offers significant completion time reduction compared with the existing algorithms.

\end{abstract}


\section{Introduction}

\subsection{Overview}
\IEEEPARstart{T}{he} exploding amount of mobile traffic, e.g., streaming applications, YouTube
videos, video-on demand,  consume large bandwidth and high transmission energy of the source limited cellular networks. Moreover, if  the cloud base-stations
(CBSs) are fully loaded,  it is not possible for the CBSs to schedule all user devices (UDs). To circumvent these challenges, D2D communications have widely been considered as a promising technology \cite{1,2,3}. The performance of D2D communications can be further improved by pushing some popular contents to the UDs near to the CBSs. This
integrated system is referred to as \textit{cache-enabled D2D} system. Cache-enabled D2D system draws remarkable benefits for alleviating the traffic congestion of the cellular
network and reducing both the CBS involvement and end-to-end latency. In this work, we consider cache-enabled D2D system, where multiple UDs cache some popular contents and cooperate among them to deliver their cached contents that are requested by other UDs. As such, all the requested contents are delivered to all UDs  within the lowest possible network
delay.

Network coding (NC) has been shown to be promising for improving  throughput and minimizing decoding delay and completion time for  numerous applications in wireless  networks \cite{2,3,4,5}. Specifically, random linear NC (RLNC) can achieve
the optimal throughput of wireless broadcast networks \cite{5}. However, this throughput achievement comes at the expense of  complex encoding (i.e., mixing contents using coefficients from a large Galois field), high decoding delay, and prohibitive
computational complexity. This is suitable for delay-tolerant contents and UDs with high capabilities and buffer sizes. The report by CISCO \cite{6} shows that a significant portion of network traffic is popular contents (popular videos and photos) that are frequently requested by UDs in short time. Therefore, it is
crucial to deliver these delay-sensitive contents with minimum possible delay. For this purpose, instantly decodable NC (IDNC) is adopted \cite{7}. IDNC performs simple encoding XOR operation at the transmitter and simple decoding XOR operation at the receiver, and thus instant use of the received contents. Accordingly, it is suitable for implementation in small and low cost UDs \cite{8}. Therefore, D2D communication and IDNC technique can be exploited to deliver popular contents to UDs with the lowest possible delay while offloading the CBSs. For instance, consider that a content consists of  set of files $f_1, f_2,$ and $f_3$ is wanted by set of UDs $u_1, u_2$, and $u_3$. Suppose that the CBS transmitted the wanted contents to the UDs and due to channel impairments UD $u_i$ did not obtain file $f_i$ for $1 \leq i \leq 3$. These missed files can be traditionally re-transmitted from the CBS to each UD until all UDs obtain them correctly. As a result, the CBS requires at least $3$ uncoded transmissions for delivering these files, which degrade system performance \cite{8a}. However, UDs can be either file cachers that can deliver their cached files to other UDs or file requesters that can receive the wanted files from other UDs. In our considered example, UD $u_1$ holds files $f_2$ and $f_3$, and accordingly, it can transmit the binary XOR combination $f_2\oplus f_3$ to UDs $u_2$ and $u_3$. Then, UD $u_2$ holds $f_1$ and can provide it to UD $u_1$. As a result, $2$ transmissions are required for delivering all files to all UDs.  Therefore, the cooperation among UDs can be utilized with IDNC to combine files and transmit them to interested UDs via D2D links. As such,  the requested files can be delivered to the requesting UDs quickly while offloading the CBS's resources. 

\subsection{Related Works and Motivation}
The content delivery problem, known as \textit{completion time minimization problem}, in IDNC-enabled networks is considered based on layer functionalities as follows. From only network-layer perspective, IDNC schedule is adopted to solve the problem in real-time applications in terms of minimizing the number of transmissions \cite{8, 9a, 9aa, 9aaa, 9}. In particular, these related works modeled the status of physical channels by file erasure probabilities and integrated such erasures in the coding decisions, e.g., see for example \cite{9a}, \cite{9}. This improves the system's performance from network-layer perspective by scheduling many UDs to the same resource block, but it degrades the performance from physical-layer perspective through selecting the minimum rates of all the scheduled UDs. This results in prolonged transmission duration and thus, consumes the time resources of network. Unlike network-layer IDNC that depends solely on file combinations for aiding the coding decisions, rate-aware IDNC (RA-IDNC) also depends on the channel capacities of different UDs. This allows a new degree-of-freedom, such as, choosing the transmitting UDs, their transmission rates, and IDNC file combinations, to optimize the content delivery problem. The authors of \cite{10, 10n, 11,12,13, 13nn, 14, 14a, 15} used RA-IDNC in centralized and decentralized networks for optimizing different system parameters. For example, the authors of \cite{12} used RA-IDNC scheme in cloud radio access networks (C-RANs) for completion time minimization. The authors of \cite{13,14} developed cross-layer IDNC to optimize throughput and Quality-of-Service (QoS) of UDs in centralized C-RANs and Fog-RANs, respectively.\ignore{ Recently, the authors in \cite{15new} developed cross-layer NC schemes for minimizing the completion time in D2D-aided Fog-RAN system.} 

For D2D systems, the authors of \cite{15} considered a vanilla-version of the completion time minimization problem. Indeed, the problem was considered by simply selecting the transmitting UD and its NC combination. However, the main drawback of the work in \cite{15} is that only one UD is allowed to transmit coded file in each transmission slot. Thus, they ignored the interference caused by different transmitting UDs to the scheduled UDs. Actually, in D2D networks, UDs are spatially distributed in a region which creates an opportunity to judiciously select multiple transmitting UDs that schedule a significant set of other UDs. Such configuration
brings a new trade-off between scheduling UDs to transmitting UDs and
choosing the coded files and the transmission rate/power. However, solving the completion time minimization problem while jointly considering the previously cached files at users, their transmission rates/powers, NC, and D2D communications has not explored yet. Furthermore, developing a joint cross-layer IDNC for the completion time minimization in D2D networks is new to the area of NC. Therefore, our setting in this work is much more realistic than the one used in \cite{15}  as it enables for both selection of multiple transmitting UDs and optimization of the employed transmission rates using power control on each transmitting UD.

The completion time minimization problem is motivated by real-time applications, i.e., video streaming. In these applications, UDs need to obtain a set of popular files from other transmitting UDs with the minimum possible completion time, given the required minimum rates for QoS.  Unlike pre-loads that can be done at much lower rates or at off-peak times, our work delivers popular contents to UDs with the minimum possible completion time. Consider that a popular video representing a frame of files is requested by a set of UDs located in a playground. Many  UDs in the playground are interested in receiving this frame.  At any given time, consider that UDs have already cached some files and requested some other files from that frame. In order to deliver the requested files in that frame without any interruption, UDs should receive their requested files with a minimum possible delay. For such a case, users can re-XOR the transmitted files from transmitting UDs to progressively and immediately use the decoded files at the application layer. Such progressive file decoding at the UDs meets the delay requirements and streaming quality. 

\subsection{Contributions}
Unlike previously-discussed existing works that considered the optimization factors (e.g., NC, users' cached and requested files, UD scheduling, QoS guarantee requirements, and power optimization) and their corresponding problems separately, our work develops a framework that jointly considers all the aforementioned factors. To this end, we develop a novel \textit{cross-layer network coding (CLNC)} optimization framework taking NC and rate/power optimization into account. The main contributions of our work are summarized as follows.
\begin{itemize}
	\item The completion time minimization problem is shown to be computationally intractable due to the interdependence
	among variables such as the UDs' cached and requested files, power optimization, channel qualities, and coding decisions. Using the lower bound on the completion time used in the literature, we tackle the problem and solve it online at each transmission slot. 
	\item We design a D2D-RA-IDNC graph to efficiently transform the completion time minimization problem to a maximal weight independent set (MWIS)	 problem. The designed D2D-RA-IDNC graph represents all the feasible rates and	NC decisions for all potential transmitting UDs. The problem is
	then reformulated as an MWIS problem that can be efficiently solved
	using low complexity graph theoretical solution. The designed weights of the vertices in this graph balances between the transmission rate and number of scheduled UDs in each transmission.
	
	\item We develop a CLNC solution that efficiently iterates between finding the MWIS in the designed D2D-RA-IDNC graph and optimizing the power of the transmitting UDs using a function
	evaluation (FE) method. In each iteration,  a new transmitting UD is selected only if the resultant interference does not significantly degrade the completion time performance. The complexity of our developed CLNC solution is analyzed.
	\item We compare our proposed scheme with existing baseline schemes. Simulation results demonstrate that our proposed CLNC solution significantly minimizes the completion time compared with existing algorithms. 
\end{itemize}

The rest of this paper is organized as follows. \sref{SMMM} overviews the system model. The completion time approximation and problem formulation are illustrated in \sref{PF}. In \sref{GC}, we present the graph construction and problem transformation and propose cross layer network coding solution in  \sref{PS}. Finally, we present selected simulation results in \sref{NC} and conclude the work in \sref{C}.

\section{System Model} \label{SMMM}
\subsection{System Overview}
We consider a cache-enabled D2D system with one cloud base-station (CBS) and $N$ user-devices (UDs), denoted by the set $\mathcal{N}=\{u_1,u_2, \cdots,u_N\}$. We adopt a fully connected D2D model where D2D links are usually implemented with low-range transmission technologies, such as Bluetooth and WiFi. Therefore, we assume that each UD is connected to all other UDs. Each UD is assumed to be equipped with single antenna and used half-duplex channel. Accordingly, each
UD can either transmit or receive at a given transmission slot. Unlike the work in \cite{15} that ideally considered interference-free setup, our realistic work considers that UDs use the same frequency band and can cooperate by utilizing D2D links and transmit simultaneously. With such a cooperation among UDs for content delivery, the CBS dose not need to transmit requested contents to UDs. Therefore, the CBS is responsible for selecting a set of transmitting UDs and their NC combinations and power allocations, that deliver requested contents to requesting UDs. Accordingly, the whole process of coding decisions in this work is executed at the CBS and depends on selecting the transmission rate and transmit power allocation of each transmitting UD.

Let $\mathcal{F}$ denote a frame of $F$ files, $\mathcal{F}=\{f_1,f_2, \cdots,f_F\}$, each of size $B$ bits. This data frame represents a popular content, i.e., YouTube video, and constitutes the set of most frequent requested files by the UDs for any given time period.\ignore{ Using UDs cooperation, this frame is required to be delivered to requesting UDs via D2D links with minimum delivery time.} We assume that UDs proactively cached some files from $\mathcal F$ and stored them in their local caches, i.e., $\mathcal C_{u_k}$ represents the set of the files locally cached at UD $u_k$. We assume
that UD $u_k$ requests a set of files, from the frame
$\mathcal F$, and denoted by the demand set $\mathcal W_{u_k}=\mathcal F\backslash \mathcal C_{u_k}$. Following the caching model in \cite{15}, each file from $\mathcal F$ is cached by at least one UD in $\mathcal N$ which leads to the fact that $\cup_{u_k \in \mathcal{N}} \mathcal C_{u_k}=\mathcal{F}$. The set of UDs that having non-empty \textit{demand} sets at the $t$-th transmission slot is denoted by $\mathcal N_{w,t}$, i.e., $\mathcal N_{w,t} = \{u_k \in \mathcal N | \mathcal{W}_{u_k,t} \neq \varnothing\}$. Without loss of generality, we assume that all UDs have non-empty \textit{demand} sets. Otherwise, they can simply
be ignored from the set $\mathcal N$ without affecting the system
performance. When an UD receives
its requested files, it acts as a transmitting UD that provides its received files to the interested UDs. The goal is to deliver the requested files to the UDs within the lowest possible completion time by leveraging NC and D2D links.

Let $\gamma_{u_k,u_i}$ denote the channel gain between UD $u_k$ and UD $u_i$ and $Q_\text{max}$ denote the maximum transmit power for
D2D link. We consider slow fading channels, and accordingly,  $\gamma_{u_k,u_i}$ is considered to be fixed  during a single transmission but may change independently from one file transmission to another file transmission.  Then, the achievable capacity of a D2D pair ($u_k, u_i$) is given by 
$C_{u_k,u_i} = \log_2(1+ \text{SINR}_{u_k,u_i}(\textbf{Q})), \forall u_k \in \mathcal A$, where $\text{SINR}_{u_k,u_i}(\textbf{Q}))$ is the corresponding signal-to-interference plus noise-ratio experienced by
UD $u_i$ when it is scheduled to UD $u_k$. This SINR is given by $\text{SINR}_{u_k,u_i}(\textbf{Q})) =  \cfrac{Q_{u_k} |\gamma_{u_k,u_i}|^{2}}{N_{0} +\sum_{u_m\neq u_k}Q_{u_m}| \gamma_{u_m,u_i}|^{2}}, \forall u_k,u_m\in \mathcal A$, where $\mathcal A$ is the set of transmitting UDs,
$N_0$ is the noise power, $Q_{u_k}$, $Q_{u_m}$ are the transmit powers of UDs $u_k$ and $u_m$ which both are bounded by $Q_\text{max}$, and $\textbf{Q} = [Q_{u_k}]$ is a row vector containing the power levels of the transmitting UDs. The channel  capacities of all pairs of D2D links  can be stored  in an $N \times N$  \emph{capacity status matrix (CSM)} $\mathbf{C} = [C_{u_k,u_i}],$ $\; \forall (u_k, u_i)$. Since UD $u_k$ cannot transmit to itself, $C_{u_k,u_k} = 0$.

\ignore{We assume that a transmitting UD $u_k$ can infinitesimally  adjust its modulation scheme to adopt any transmission rate $r(t)$. For transmission synchronization, all devices in $\mathcal A$ should  adopt a common transmission rate $r(t)$ or there will be a waiting time to synchronize all devices.
	Moreover, the $t$-th D2D transmission from device $d_k$ will be successful  at  device $d_i\in \mathcal C_k$ if the adopted  transmission rate $r(t)$ is less than or equal to  the channel  capacity $C_{k,i}$ (i.e.,  $r(t) \leq C_{k,i}$).\footnotemark \footnotetext{ We focuse only in this work on  power control and rate adaptation aware IDNC for D2D networks. The transmission loss due to channel estimation errors, device mobility and hidden terminal collisions is not considered and left as an intersting future work.}  If $r(t) >  C_{k,i}$,  device $d_i$ will not be able to receive the transmission from device $d_k$.}

\subsection{Rate-Aware NC and Expression of the Completion Time Metric}
Let $\mathtt {f}_{u_k,t}$ denote the XOR file combination to be sent by UD $u_k$ to the set of scheduled UDs $\mathtt u(\mathtt {f}_{{u_k},t})$ at the $t$-th transmission. For simplicity, we use time index $t$  to represent the $t$-th transmission slot, i.e., $t = 1$ refers to the first transmission slot. The file combination $\mathtt {f}_{u_k}$ is an element of the power set $\mathcal P({\mathcal C_{u_k}})$ of the cached files at UD $u_k$. At every transmission $t$, each scheduled UD  in $\mathtt u(\mathtt {f}_{{u_k,t}})$ can re-XOR $\mathtt {f}_{{u_k,t}}$  with its
previously cached files to decode a new requested file. To ensure successful reception at the UDs, the  maximum transmission rate of a particular transmitting UD  is equal to the minimum achievable capacity of its scheduled
UDs. Therefore, the set of targeted users\footnote{The term ``targeted users" is given for the scheduled UDs who receive an instantly-decodable transmission.} by UD $u_k$ is expressed as $
\mathtt {u}(\mathtt {f}_{u_k})=\left\{u_i \in \mathcal{N}_w \ \big||\mathtt {f}_{u_k} \cap \mathcal{W}_{u_i}| = 1~\text{and}~ R_{u_k} \leq C_{u_k,u_i} \right\}$.
Without loss of generality, the set of all targeted UDs, when $|\mathcal A|$ transmitting UDs transmit the set of combinations $\mathtt {f}(\mathcal A)$, is represented by $
\mathtt {u}(\mathtt {f}(\mathcal A))$, where $u_k$, $\mathtt {f}_{u_k}$, $\mathtt {u}({\mathtt f}_{u_k})$ are elements in $\mathcal A$, $\mathtt {f}(\mathcal A)$, and $\mathtt {u}(\mathtt {f}(\mathcal A))$, respectively \footnote{The symbol $|\mathcal X|$ represents the cardinality of the set $\mathcal X$.}.

Let $T_{u_k}$ denote the duration of the transmission from the $u$-th UD.  The duration for transmitting $\mathtt {f}_{u_k}$ from UD $u_k$ with rate $R_{u_k}$ to $\mathtt u(\mathtt f_{u_k})$ is $T_{u_k}=\frac{B}{R_{u_k}}$ seconds.
For transmission synchronization, all transmitting UDs in the set $\mathcal A$ adopt a common
transmission rate, denoted as $R$. Otherwise, different transmitting UDs will have different transmission rates, and thus, they will have different transmission
durations. So, UDs who finish the transmission first must wait for those who transmit with the slowest-rate to start a new transmission at the same time. Therefore, we adopt one transmission rate for all transmitting UDs, and accordingly, the transmission duration for sending any coded/uncoded file from any transmitting UD is denoted by $T_t$ and expressed by $T_{t}=\frac{B}{R}$ seconds. Consequently,
UDs that are not targeted at transmission slot $t$, 
experience $T_{t}$ seconds of delay has a cumulative delay as defined below.\\
\textit{\textbf{Definition 1:} Any UD with non-empty demand set experiences $T_{t}$ seconds of time delay if it does not receive any requested file at  $t$-th transmission. The accumulated time delay of UD $u_i$ is the sum of $T_{t}$ seconds at each transmission until $t$-th transmission, denoted by $\mathbb{T}_{u_i}(t)$, and expressed as 
	\begin{align} \label{eq3}
	\mathbb{T}_{u_i}(t) = \mathbb{T}_{u_i}(t-1)+
	\begin{cases}
	T_t & \text{if} ~u_i \notin  \mathtt u{(\mathtt f(\mathcal A))}\\
	T_t & \text{if} ~u_i\in \mathcal A.
	\end{cases}
	\end{align}}

Let $\mathtt T_{u_i}$ denote the completion time of UD $u_i$ until it receives the requested files. The completion time for UD $u_i$ includes two parts, its accumulated time delay $\mathbb{T}_{u_i}$ due to receiving a non-instantly decodable file and the time duration of sending all instantly decodable transmissions. In other words, such completion time is divided into consecutive instantly and non instantly decodable transmission for each UD in $\mathcal N_w$ until it obtains all requested files. Subsequently, the overall completion time $\mathtt T = \max_{u_i \in \mathcal N} \{\mathtt T_{u_i}\}$ is the time required until all UDs recover all files. The used notations and variables are summarized in \tref{table_1}.

To minimize the overall completion time, we need to find the optimal schedule from the beginning of the D2D transmission phase at $t=1$ until all UDs obtained all requested files at $t=\mathcal {j{\mathcal{S}}j}$. Here,  $\mathcal{S}$ is defined  as a collection  of  transmitting UDs, file  combinations   and transmission rates/powers  until all UDs in $\mathcal N_w$ receive all $F$ files, i.e., $ \mathcal{S}= \{\mathcal A(t), \mathcal P(\mathcal C_{u_i}(t)), \mathcal R(t)\}, \forall t \in \{1,...,|\mathcal S|\}$. Thus, the optimal  schedule $\mathcal S^*$ that minimizes the overall completion time of all UDs is
$\mathcal S^* = \arg\min_{\mathcal S \in \mathbf{S}}\{\mathtt T(\mathcal S)\} = \arg\min_{\mathcal S \in \mathbf{S}} \left \{\max_{u_i \in \mathcal N_w}\left\{ \mathtt T_{u_i}(\mathcal S)\right\}\right \}$, where $\mathbf{S}$ is the set of all possible D2D transmission schedules, i.e.,    $\mathcal{S} \in \mathbf{S}$. This optimal schedule can be formulated as follows

\begin{table}[t!]
	\renewcommand{\arraystretch}{0.9}
	\caption{Variables and parameters of the system}
	\label{table_1}
	\centering
	\begin{tabular}{|p{1.5cm}| p{5.9cm}|}
		\hline
		Variable & Definition\\
		\hline
		\hline
		$\mathcal{N}$ & Set of $N$ UDs\\
		\hline
		$\mathcal{N}_w$ & Set of $N$ UDs that want files\\
		\hline
		$\mathcal{F}$ & Set of $F$ popular files\\
		\hline
		$B$ & File size \\
		\hline
		$\mathtt f_{u_k}$ & The encoded file of of UD $u_k$\\
		\hline 
		$\mathtt u(\mathtt f_{u_k})$ & Set of targeted UDs by UD $u_k$\\
		\hline 
		$\mathcal{A}$ & Set of $A$ transmitting UDs\\
		
		\hline
		$\textbf{C}$ & Set of all achievable capacities\\
		\hline
		$\mathcal{W}_{u_i}$ & Set of wanted files by UD $u_i$\\
		\hline
		$\mathcal{C}_{u_i}$ & Set of locally cached files by UD $u_i$\\
		\hline
		$R_{u_k}$ & Transmission rate of UD $u_k$\\
		\hline
		$T_{u_k}$ & The transmission duration of UD $u_k$\\
		\hline
		$\mathtt T_{u_i}$ & The completion time of UD $u_i$\\
		\hline
		$\mathcal{R}_{u_k}$ & Set of all achievable rates of UD $u_k$\\
		\hline

	\end{tabular}
\end{table}

\begin{theorem}\label{th:pmp}
	The minimum  overall completion time  problem  in a D2D multihop  network can be formulated as a transmission schedule selection problem such that:
	\begin{align} \label{eq:Sopt}
	\mathcal S^* = \arg\min_{\mathcal S \in \mathbf{S}} \left \{\max_{u_i \in \mathcal N_w}\left \{\frac{B.|W_{u_i}(0)|}{\tilde{R}_{u_i}(\mathcal S)} + \mathbb T_{u_i}(\mathcal S)\right \} \right \},
	\end{align}
	where $|W_{u_i}(0)|$ is the initial demand set size of UD $u_k$,  $\mathcal \mathbb T_{u_i}(\mathcal S)$ is the accumulative time delay of UD $u_i$ in schedule $\mathcal S$ and   $\tilde{R}_{u_i}(\mathcal S)$  is the harmonic mean of the  transmission rates of  time indices   that   are instantly decodable for UD   $u_i$ in schedule $\mathcal S$.
\end{theorem}

\begin{proof}
	The proof of  \thref{th:pmp} is omitted in this paper becasue we can use  the same steps that was used in  \cite{11} for C-RAN networks. Therefore, a sketch of the proof is given as follows. We first show that the completion time can be expressed as the sum of instantly and non-instantly decodable transmission times from $|\mathcal A|$ transmitters via  D2D links. Afterward, we need to proof that the number of instantly decodable transmissions to UD $u_l$ is equal to the number of its requested files $|\mathcal W_{u_l,0}|$ and
	the number of non-instantly decodable transmissions matches
	the time delay in definition 1. Finally, we extend the results of the optimal schedule in Theorem
	1 in \cite{11} that was used in C-RAN system to the coordinated D2D setting with multiple transmitters.
\end{proof}

The optimal NC transmission schedule that reduces the overall completion time in a D2D network is the solution of the optimization problem in \thref{th:pmp}. Such schedule requires to exploit the heterogeneity of UDs’ channel capacities and the interdependence of UDs’ file reception. Actually, the decision at the current transmission slot is dependent on the future coding situations, which makes the optimization problem anti-causal. Therefore, it can be inferred that finding the optimal schedule $\mathcal S^*$ is intractable \cite{12}, \cite{15}.

\begin{figure}[t]
	\centering
	\includegraphics[width=1\linewidth]{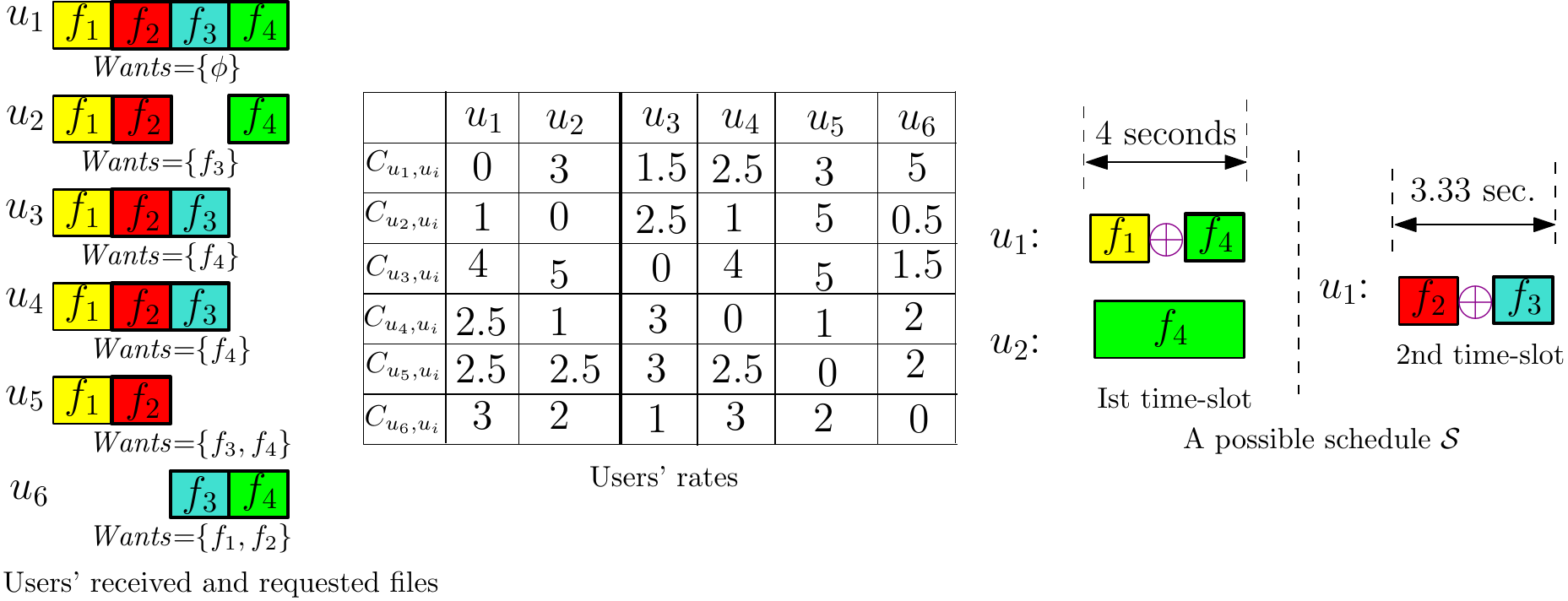}
	\caption{D2D system containing $6$ UDs and their corresponding requested/received files and rates. For example, $u_2$ receives $f_1$, $f_4$ and requests $f_2$, $f_3$. The sets of files that locally cached at UD $u_1$ is: $\mathcal C_{u_1}=\{f_1, f_4, f_3\}$.}
	\label{fig33}
\end{figure}

\subsection{Example of RA-IDNC Transmissions in D2D System} This example illustrates the aforementioned definitions and concepts to ease the analysis of the completion time minimization problem reformulation in next section. Consider a simple D2D network that shown in  Fig. \ref{fig33} which consists of $6$ users, users' received and requested files and their rates.  For example, $u_2$ receives $f_1$, $f_4$ and requests $f_2$, $f_3$. Each file is assumed to have a size of $10$ bits. To minimize the completion time for this example, one possible schedule is given as follows.\\
\textbf{First time slot:} The $u_1$-th and $u_2$-th UDs can use their cached files to transmit $\mathtt f_{u_1}=f_1 \oplus f_4$ and $\mathtt f_{u_2}=f_4$ with rates $R_{u_1}=2.5$ and $R_{u_2}=2.5$ bits/s, respectively, to the sets  $\mathtt u(\mathtt f_{u_1})=\{u_4,u_6\}$ and $\mathtt u(\mathtt f_{u_2})=\{u_3,u_5\}$. Given this, we have the following transmission durations of $u_1$-th UD and $u_2$-th UD, respectively: $T_{u_1}=\frac{10}{2.5} = 4, T_{u_2}= \frac{10}{2.5} = 4$ seconds. The decoding process at UDs side can be explained as follows.
\begin{itemize} 
			\item The $u_4$-th UD already has $f_1$, so it can XOR the combination ($f_1\oplus f_4$) with $f_1$ (i.e., $(f_1\oplus f_4)\oplus f_1$) to retrieve $f_4$. Thus, the transmission is instantly decodable for UD $u_4$.
			
			\item The $u_6$-th UD already has $f_4$, so it can XOR the combination ($f_1\oplus f_4$) with $f_4$ (i.e., $(f_1\oplus f_4)\oplus f_4$) to retrieve $f_1$. Thus, the transmission is instantly decodable for UD $u_6$.
			
			\item The $u_3$-th and $u_5$-th UDs can receive $f_4$ from $u_2$-th UD. Thus, the transmission is instantly decodable for UDs $u_3$ and $u_5$.
			
\end{itemize}	
Therefore, the updated \textit{demand} sets after the first time slot are: $\mathcal W_{u_2} = \{f_3\}$, $\mathcal W_{u_3} = \varnothing, \mathcal W_{u_4} = \varnothing$, $\mathcal W_{u_5} = \{f_3\}$, $\mathcal W_{u_6} = \{f_2\}$. Note that $T_{t,1}= 4$ seconds.\\
\textbf{Second time slot:} The $u_1$-th UD can use its cached files to transmit $\mathtt f_{u_1}=f_2 \oplus f_3$ with rate $R_{u_1}=3$ bits/s to the set $\mathtt u(\mathtt f_{e_2})=\{u_2,u_5, u_6\}$ which requires transmission time  $T_{u_1}=T_{t,2}=\frac{10}{3} = 3.33$ seconds.
The decoding process at UDs side can be explained as follows.
\begin{itemize} 
	\item The $u_2$-th UD already has $f_2$, so it can XOR the combination ($f_2\oplus f_3$) with $f_2$ (i.e., $(f_2\oplus f_3)\oplus f_2$) to retrieve $f_3$. Thus, the transmission is instantly decodable for UD $u_2$.
	
	\item The $u_5$-th UD already has $f_2$, so it can XOR the combination ($f_2\oplus f_3$) with $f_3$ (i.e., $(f_2\oplus f_3)\oplus f_2$) to retrieve $f_3$. Thus, the transmission is instantly decodable for UD $u_5$.
	
		\item The $u_6$-th UD already has $f_2$, so it can XOR the combination ($f_2\oplus f_3$) with $f_2$ (i.e., $(f_2\oplus f_3)\oplus f_2$) to retrieve $f_3$. Thus, the transmission is instantly decodable for UD $u_6$.
	
\end{itemize}	
By the end of second time slot, all UDs will have their requested files. Therefore, the total transmission time is $ T_{t,1} + T_{t,2}=4+3.33=7.33$ seconds.	

\ignore{
Let the transmission schedule $\mathcal{S}_2$ of the work proposed in [16] be:
	\begin{enumerate}
		\item \textbf{First time slot:} eRRH$_1$ and eRRH$_2$ adopt the transmission rate to $2.5$ bits/s to transmit coded  file $\mathtt f^c_1=f_1 \oplus f_3$ to serve users $6$, $5$, and transmit uncoded  file $f_2$ to serve user $2$, respectively. User $1$ can send $f_4$ to user $4$ with the same adopted eRRHs' transmission rate of $2.5$ bits/s. The  transmission time is $T^{c}_1=T^{c}_2=T^{d2d}_{1}=\frac{10}{2.5} = 4$ seconds. Subsequently, the \textit{Wants} sets are: $\mathcal W_2 = \{f_3\}$, $\mathcal W_3 = \{f_4\}, \mathcal W_4 = \{\varnothing\}$, $\mathcal W_5 = \{f_4\}$, $\mathcal W_6 = \{f_2\}$.
		
		\item \textbf{Second time slot:} eRRH$_1$ and eRRH$_2$ adopt the transmission rate to $1.5$ bits/s to transmit coded  file $\mathtt f^c_1=f_2 \oplus f_4$ to serve  users $3$, $5$, $6$, and transmit uncoded file $f_3$ to serve  user $2$, respectively. The  transmission time $T_2$ is $\frac{10}{1.5} = 6.666$ seconds. Subsequently, all  users have received their wanted files.
	\end{enumerate}
Schedule $\mathcal{S}_2$ requires total transmission time of $ T^1_{\max}+ T^2_{\max}= 10.666$ seconds. With all these results, it can be concluded that schedule $\mathcal{S}_1$ minimizes the completion time in seconds compared to schedule $\mathcal{S}_2$.}
The above example demonstrates the benefit of NC and D2D communications in minimizing the completion time. We can  further improve this result by allocating the power levels efficiently to the transmitting UDs.

\section{Completion Time Approximation and
	Problem Reformulation}\label{PF}
Following \cite{11, 15}, we approximate the completion time in \thref{th:pmp}  to select  a set of transmitting UDs, file combinations, and  transmission rates/powers at each transmission slot $t$ without going through all future possible  coding decisions. To achieve this, at each transmission slot $t$, a lower bound on the  completion times  of all UDs is computed. This lower bound is computed separately for each UD and does not require to exploit  the interdependence of UDs' file reception and channel capacities. In fact, this lower bound metric facilitates the  mapping of   the transmission schedule selection problem in \eqref{eq:Sopt} into   an online maximal independent set selection problem.

\begin{corollary} \label{cor:LBcompletion}
	A  lower bound on completion time $\bar{\mathtt T}_i(t)$ of UD $u_i \in \mathcal N_w$ in a given  time index $t$ can be approximated  as
	\begin{align}\label{eq:Aig}
	\bar{\mathtt T}_{u_i}(t) \approx  \frac{B.|W_{u_i}(0)|}{\tilde{R}_{u_i}}  + \mathbb T_{u_i}(t),
	\end{align}
	where $\mathbb T_{u_i}(t)$ is the accumulative time delay experienced by  UD $u_i$ until time index $t$ and  $\tilde{R}_{u_i}$ is the harmonic mean of the channel  capacities from all UDs.
\end{corollary}

\begin{proof}
	The expression in  \eref{eq:Aig} matches the expression  in \thref{th:pmp}, except $\mathbb T_{u_i}(\mathcal S)$ and $\tilde{R}_{u_i}(\mathcal S)$ of  \thref{th:pmp} are replaced by $\mathbb T_{u_i}(t)$ and $\tilde{R}_{u_i}$, respectively. The best case scenario is that  all   transmissions starting from  time slot $t$ are instantly decodable for  UD  $u_i$. Thus, it  experiences no further time delay, i.e., $\mathbb T_{u_i}(\mathcal S) = \mathbb T_{u_i}(t)$. In addition, since a fully connected D2D model is adopted, UD $u_i$ can receive  a missing file  from any other UD until it receives all $F$ files. Therefore,  $\tilde{R}_{u_i}(\mathcal S)$ is replaced  by $\tilde{R}_{u_i}$, where $\tilde{R}_{u_i}$ is the harmonic  mean of the channel  capacities from all other  UDs to UD $u_i$. This is an approximation as $\tilde{R}_{u_i}$ is exactly equal to $\tilde{R}_{u_i}(\mathcal S)$ if UD  $u_i$ receives an equal number of files from other UDs with the rates of channel capacities.
\end{proof}

Using the approximated completion time \eref{eq:Aig} at each transmission slot $t$, we now ready to reformulate the completion time minimization problem in \thref{th:pmp} with the aim to develop a
cross-layer network coding framework that decides the set of transmitting UDs $\mathcal A$ for sending $\mathtt f_{u_k}$ to the UDs $\mathtt u(\mathtt f_{u_k})$, and their transmission
rate/power $\{R_{u_k}, Q_{u_k}\}$, $\forall u_k\in \mathcal A$. As
such, all files are delivered to all UDs with minimum completion time. Therefore, the completion time minimization problem in fully D2D connected system can be formulated as
\begin{subequations}
	\begin{align} \label{eqn:Mschedule}
	&\rm P1: \hspace{0.2cm} \min_{\substack{\mathtt f_{u_k}, r_{u_k}, Q_{u_k}\\ \mathcal  A\in \mathcal P(\mathcal{N})
	}} \left \{\max_{u_i \in \mathcal N_w} \bar{\mathtt T}_{u_i}(t)\right \}\\ \nonumber
	&\rm subject~to 
	\begin{cases}
	\text{(C1):}\hspace{0.1cm} \mathtt u(\mathtt f_{u_k}) \cap \mathtt u(\mathtt f_{u_m}) =\varnothing, \forall u_k \neq {u_m} \in \mathcal A,\\
	\text{(C2):}\hspace{0.1cm} \mathtt f_{u_k}\subseteq \mathcal P( \mathcal{H}_{u_k}), ~\forall u_k\in \mathcal A,  \\
	\text{(C3):}\hspace{0.1cm} 0\leq Q_{u_k}\leq Q_{\max}, ~\forall u_k\in \mathcal A,\\ \text{(C4):}\hspace{0.1cm} R_{u_k}\geq R_{\text{th}},  \forall u_k\in \mathcal A,
	\end{cases}
	\end{align}
\end{subequations}
where (C1) states that the sets of targeted UDs from all transmitting UDs
are disjoint, i.e., each UD must be scheduled to only one transmitting UD; (C2) ensures that all files to be combined using
XOR operation at each transmitting UD $u_k$ are stored in its \textit{cache}; (C3) bounds the maximum transmit power of transmitting UDs, and (C4) guarantees the minimum transmission rate $R_\text{th}$ required to meet the QoS rate requirements.

The optimization variables in (\rm P1) contain the NC scheduling parameters  $\mathtt u(\mathtt f_{u_k})$,  potential set of transmitting UDs $\mathcal A$, and their  adopted power allocations. It can be seen that problem (\rm P1) is intractable. However, by analyzing the problem, next section successfully transforms it into MWIS problem  using graph theory technique.

\section{Graph Construction and Problem Transformation}\label{GC}
The formulated problem in  ($\rm P1$) is similar to 
MWIS problems in several aspects. In MWIS, two vertices should be non-adjacent in the graph,
and similarly, in problem ($\rm P1$), same UD
cannot be scheduled to two different UDs (i.e., C1). Moreover, the objective of problem ($\rm P1$) is to minimize the maximum completion time, and similarly, the goal of MWIS is to maximize the number of vertices that have high weights. Therefore, the feasible NC schedules can be considered to be the MWISs. Consequently, we focus on graph-based methods, and in what follows, we will construct a graph that allows us to transform problem ($\rm P1$) into MWIS-based problem.

\subsection{D2D Rate-Aware IDNC Graph} \label{Graph}
In this sub-section, we construct a weighted undirected graph, referred to D2D-RA-IDNC graph, that considers all possible conflicts for scheduling UDs, such as transmission, network coding, and transmission rate. Let $\mathcal{G}(\mathcal V, \mathcal E)$ represent the D2D-RA-IDNC graph where $\mathcal V$, $\mathcal E$ stand for the set of all the vertices and the edges, respectively. In order to construct $\mathcal{G}$, we need first to generate the vertices and connect them.

Let $\mathcal N_{w}\subset \mathcal N$ denote the set of UDs that still wants some files. Hence, the D2D-RA-IDNC graph is designed by generating all vertices for the $u_k$-th possible transmitting UD, $\forall u_k \in \mathcal N$. The vertex set $\mathcal V$ of the entire graph is the union of vertices of all UDs. Consider, for now, generating the vertices  of UD $u_k$. Note that transmitting UD  $u_k$ can encode its IDNC file $\mathtt f_{u_k}$ using  its previously received files $\mathcal C_{u_k}$. Therefore, each
vertex is generated for each single file $f_h\in \mathcal W_{u_i}\cap \mathcal C_{u_k}$ that is requested by each UD $u_i\in \mathcal N_{w}$ and for each achievable rate $r$ of UD $u_k$ that is defined below.\\
\textit{\textbf{Definition 2:}
	The set of achievable rates $\mathcal R_{u_k,u_i}$ from UD $u_k$ to UD $u_i$ is  a subset of achievable rates $\mathcal R_{u_k}$ that are less than or equal to channel  capacity  $r_{u_k,u_i}$. It can be expressed by  $\mathcal R_{u_k,u_i} = \{r \in \mathcal R_{u_k}| r \leq C_{u_k,u_i}~ \text{and} ~u_i\in \mathcal N_\text{w}\}$.}

The above definition emphasizes that $u_i$-th UD can  receive a file from transmitting UD $u_k$ if the adopted  transmission rate $r$ is in the achievable set $R_{u_k,u_i}$. Therefore, we generate $|\mathcal R_{u_k,u_i}|$ vertices for a requesting  file  $f_h \in \mathcal C_{u_k}  \cap \mathcal W_{u_i}, \forall u_i \in  \mathcal N_\text{w}$. In summery,  a vertex $v_{r,i,f}^k$ is generated  for each association of transmitting UD $u_k$, a  rate  $r \in \mathcal R_{u_k,u_i}$, and  a requesting file  $f_h \in \mathcal C_{u_k}  \cap \mathcal W_{u_i}$ of user $u_i \in  \mathcal N_\text{w}$. Similarly, we generate all vertices  for all UDs in $\mathcal N$.

Given the above generated vertices, in what follows, we connect them to construct the D2D-RA-IDNC graph. All possible conflict connections  between vertices (conflict edges between circles) in the D2D-RA-IDNC graph are provided as follows. Two vertices $v_{r,i,h}^k$ and $v_{r',i',h'}^{k}$ representing the same
transmitting UD $u_k$ are linked with a coding-conflict edge if the resulting
combination violate the instant decodability constraint.
This event occurs if one of the following holds.
\begin{itemize}
	\item The combination is not-instantly decodable, i.e., $f_h\neq f_{h^\prime}$ and ($f_{h},f_{h^\prime}$) $\notin \mathcal C_{u_{k^\prime}}\times \mathcal C_{u_{k}}$.
	\item The transmission rate is different, i.e., $r\neq r'$.
\end{itemize}

Similarly, two vertices $v_{r,i,h}^k$ and $v_{r',i',h'}^{k'}$ representing different transmitting UDs $u_k\neq u_{k'}$ are conflicting if

\begin{itemize}
	\item The transmission rate is different, i.e., $r\neq r'$.
	\item The same UD is scheduled, i.e., $u_i=u_{i^\prime}$.
\end{itemize}

Therefore, two vertices $v_{r,i,h}^k$ and $v_{r',i',h'}^{k'}$ are adjacent
by a conflict edge in $\mathcal E$ if they satisfy one of the following connectivity conditions (CC).
\begin{itemize}
	\item  CC1: $u_k=u_{k^\prime}$ and ($f_h\neq f_{h^\prime}$) and ($f_{h},f_{h^\prime}$) $\notin \mathcal C_{u_{k^\prime}}\times \mathcal C_{u_{k}}$.
	\item  CC2:  $r\neq r'$. 
	\item  CC3:  $u_k \neq u_{k^\prime}$ and $u_i=u_{i^\prime}$. 
\end{itemize}

\subsection{Problem Transformation}\label{PT}
In this sub-section, we transform the network-coded user scheduling and power optimization problem ($\rm P1$) into MWIS problem, and consequently, we start by the following definitions.\\
\textit{\textbf{Definition 3:} Any independent set (IS) $\mathcal I$ in graph $\mathcal G$ must satisfies: i) $\mathcal I_i \subseteq \mathcal G$; ii) $\forall v, v' \in \mathcal I_i$, we have $(v, v') \notin \mathcal E$.}\\
\textit{\textbf{Definition 4:} A maximal IS in an undirected graph cannot be expanded to	add one more vertex without affecting the pairwise non-adjacent vertices.}\\
\textit{\textbf{Definition 5:} The independent set $\mathcal I$ is referred to as an MWIS of $\mathcal G$ if it satisfies: i) $\mathcal I $ is an IS in graph $\mathcal G$; ii) the sum weights of the vertices in $\mathcal I$ offers the maximum among all ISs of $\mathcal G$. Therefore, the MWIS will be denoted as $\mathtt I$.}\\
Based on the aforementioned designed D2D-RA-IDNC graph and definition of MWIS, we have the following proposition.
\begin{proposition} \label{thm:Mweight1}
The problem of minimizing the approximated completion time in ($\rm P1$) at the $t$-th transmission  is equivalently represented by the MWIS selection  among all the ISs
in the $\mathcal G$ graph, where the original weight $\omega_o(v)$ of each vertex $v$ is given by
\begin{align} \label{eqw}
	\omega_o(v)=2^{N_w-d_{u_i}+1}\bar{\mathtt T}_{u_i}(t)\left(\frac{r}{B}\right),
	\end{align}
	where $d_{u_i}$ is the  order of UD $u_k$  in the group that arranges  all  UDs in $\mathcal N_w(t)$ in decreasing order of  lower bound on completion times \cite{15}.
\end{proposition}
\begin{proof}  The proof of the proposition follows similar steps of \cite{11}, and consequently, the detailed steps are omitted. Accordingly, a sketch of proof is provided as follows. First, we need to sufficiently show that there is a mapping between the set of maximal ISs in the D2D-RA-IDNC graph and the set of feasible transmissions. Then, the weight of each IS is the objective function to P1. The authors in \cite{11} showed that there exists a one-to-one mapping between the set of feasible transmissions and the set of ISs in the RA-IDNC graph. Here, we extend the results of \cite{11} to the D2D-RA-IDNC graph by showing that the feasible transmissions between different transmissions are non-adjacent, i.e., the constraint CC3. Since each feasible transmission by a transmitting UD is an IS and they are non-adjacent, then the union of both sets is also an IS. From CC3, the same UD cannot be targeted by distinct transmitting UDs. Therefore, all vertices in the sub-graph representing transmitting UD $u_k$ are non-adjacent to vertices in the sub-graph of transmitting UD $u_{k^\prime}$ as long as the targeted UDs are distinct. Therefore, each feasible association between targeted UDs-transmitting UDs, file combinations, and the transmission rate is represented by a maximal IS. Conversely, it can readily be seen that each IS represents a feasible condition as it does not violate the connectivity conditions CC1, CC2, and CC3. Indeed, for $\mathcal{I}$, the transmission of the combination $\mathtt u_{u_k} = \oplus_{v^k_{r,i,h} \in \mathcal{I}} f$ by transmitting UD $u_k$ at rate $r$ is instantly for all UDs $\mathtt u (\mathtt f_{u_k}) = \cup_{v^k_{r,i,h} \in \mathcal I} u$.
	
To finish the proof, we show that the weight of the IS is the objective function of ($\rm P1$). Let the weight of vertex $v^k_{r,i,h}$ be defined as in \eref{eqw} and $\mathcal I$ be the set of maximal ISs in the D2D-RA-IDNC graph $\mathcal G$. Consider $\mathtt I \in \mathcal I$ is the MWIS that has the maximum vertex weights. By the designed graph $\mathcal G$,  all the feasible decisions of transmitting UDs, transmitted files and transmission rates/powers are mapped to the set of all maximal ISs. Consequently, the completion time reduction problem can be reformulated as a maximal IS selection problem in graph $\mathcal G$ such as	
	\begin{align} \label{eq:Vwei}
	&\arg \max_{\substack{\mathcal A\in \mathcal P(\mathcal N)\ignore{\\ u_k \in \mathcal A} \\ \mathtt f_{u_k} \in \mathcal{P}(\mathcal{C}_{u_k}) \\ Q_{u_k}\in\{0, \cdots, Q_{\max}\}\\ r \in \mathcal{R}_{u_k}}} \sum_{u_i \in \mathcal X} 2^{N_w-d_{u_i}+1}\bar{\mathtt T}_{u_i}(t)\left(\frac{r}{B}\right) \nonumber \\& = \max_{\substack{\mathtt I\in \mathcal I}}\sum_{v \in \mathtt I} 2^{N_w-d_{u_i}+1}\bar{\mathtt T}_{u_i}(t)\left(\frac{r}{B}\right)= \max_{\substack{\mathtt I\in \mathcal I}}\sum_{v \in \mathtt I}\omega_o(v).
	\end{align}
Consequently, the problem of choosing transmitting UDs, file combinations, and transmission rates/powers that results in minimizing the completion time is equivalent to the MWIS selection problem over the D2D-RA-IDNC graph.
\end{proof} 

It is readily known that finding the MWIS is NP-complete problem \cite{16}. Consequently, solving \pref{thm:Mweight1} is NP-hard. In the next section, we greedily select a maximal IS using the vertices’ weights defined in \eref{eqw}.

\section{Proposed Solution}\label{PS}
In this section, we develop an efficient cross-layer network coding solution  that judiciously selects multiple transmitting UDs simultaneously and their coding decisions and transmitting rates/powers. As shown in SINR expression,  the increase in the number of transmitting UDs also increases interference of a transmission channel caused by multiple transmitting UDs and therefore, reduces the  channel capacity. To control the deleterious impact of interference on channel capacities, a power allocation mechanism is employed that efficiently selects the set of transmitting UDs and  allocates the transmitting power to the transmitting UDs such that: (1) a potential number of UDs can be targeted with an IDNC combination, and (2) the channel capacities of the transmitting UDs to the targeted UDs still improves the objective function. The overall steps of our proposed solution are as follows. We first present a power allocation algorithm for the given set of transmitting UDs and the scheduled/targeted UDs to these transmitting UDs. Next, we provide a greedy algorithm that selects a set of transmitting and targeted UDs considering known/predefined power allocations. Finally, by combining the aforementioned algorithms, we present an innovative cross-layer NC solution.

\subsection{Transmit Power Allocation Algorithm}
In this sub-section, we derive optimal power allocations to maximize sum-throughput for a given set of transmitting UDs. We assume that the system has $A$ transmitting UDs, i.e., $\mathcal A=\left\{1,2, \cdots,A\right\}$ and the UDs receiving data from the $u_k$-th transmitting UD is denoted by the set $\mathtt u(\mathtt f_{u_k})$. The power optimization problem to maximize the sum-capacity of $A$ transmitting UDs is formulated as
\begin{equation}
\label{Power_opt}
\begin{split}
\max_{\{Q_k\}} \sum_{k=1}^A \mathcal N_k \quad \text{s.t.} \quad 0 \leq Q_k \leq  Q_{\max}, \forall k
\end{split}
\end{equation}
where $\mathcal N_k=\sum_{ u_i \in \mathtt u(\mathtt f_{u_k})} \log_2\left(1+\text{SINR}_{u_k,u_i}\right)$. The near-optimal power allocation for the $u_k$-th transmitting UD is obtained in the following proposition.

\begin{proposition}
	\label{Power_Prop}
	Let $\widehat{Q}_k$ be the given transmit power of the $k$-th transmitting UD at the $t$-th iteration. A converged power allocation is obtained by updating power at the $(t+1)$-th iteration, $\forall t$, according to the following power update equation
	\begin{align} 
	\label{Power_update}
	Q_k=\left[\frac{\sum_{ u_i \in \mathtt u(\mathtt f_{u_k})}\frac{\text{SINR}_{u_k,u_i}}{1+\text{SINR}_{u_k,u_i}}}{\sum_{{\substack{m=1\\ m\neq k} }} \sum_{ u_j \in \mathtt u(\mathtt f_{u_m})}\left(\frac{\text{SINR}_{u_m,u_j}}{1+\text{SINR}_{u_m,u_j}}\right)^2\frac{\gamma_{u_k,u_i}}{\widehat{Q}_m\gamma_{u_m,u_j}}}\right]_{0}^{Q_{max}}
	\end{align}
\end{proposition}
where $\text{SINR}_{u_m, u_j}$, $\forall m,j$, is obtained by applying the value $\widehat{Q}_m$ in the expression of end-to-end SINR.

\begin{proof} The proof follows similar steps of \cite[Lemma 2]{Ahmed_Multi_level_EC}. Particularly, although \eqref{Power_opt} is a non-convex power allocation problem, a local optimal solution to \eqref{Power_opt} can be obtained by obtaining the stationary point of the objective function. To obtain a stationary power allocation for the $u_k$-th transmitting UD, we need to solve $\frac{\partial \mathcal{N}_k}{\partial Q_k}=0$. In particular, we obtain
\begin{align}
\label{power_opt_2}
&\frac{\partial \mathcal{N}_k}{\partial Q_k}=\frac{1}{Q_k} \sum_{ u_i \in \mathtt u(\mathtt f_{u_k})}\frac{\text{SINR}_{u_k,u_i}}{1+\text{SINR}_{u_k,u_i}}\\ \nonumber
&-\sum_{m=1, m\neq k} \sum_{ u_j \in \mathtt u(\mathtt f_{u_m})} \left(\frac{\text{SINR}_{u_m,u_j}}{1+\text{SINR}_{u_m,u_j}}\right)^2 \frac{\gamma_{u_k,u_i}}{\widehat{Q}_m\gamma_{u_m,u_j}}.
\end{align} 
Therefore, by solving $\frac{\partial \mathcal{N}_k}{\partial Q_k}=0$, we obtain
\begin{equation}
\label{power_opt_3}
Q_k=\left[\frac{\sum_{ u_i \in \mathtt u(\mathtt f_{u_k})}\frac{\text{SINR}_{u_k,u_i}}{1+\text{SINR}_{u_k,u_i}}}{\sum_{m=1, m\neq k} \sum_{ u_j \in \mathtt u(\mathtt f_{u_m})}\left(\frac{\text{SINR}_{u_m,u_j}}{1+\text{SINR}_{u_m,u_j}}\right)^2\frac{\gamma_{u_k,u_i}}{{Q}_m\gamma_{u_m,u_j}}}\right]
\end{equation}
By solving \eqref{power_opt_3}, one can obtain the stationary point for the objective function of the $u_k$-th transmitting UD, $\forall u_k$. However, a closed-form power allocation by solving \eqref{power_opt_3} is intractable. Accordingly, we adopt an iterative approach to obtain a near-optimal stationary power allocation. To this end, we denote $\widehat{Q}_k$ as the given power allocation for the $u_k$-th transmitting UD, $\forall u_k$, and evaluate the R.H.S. of \eqref{power_opt_3} for the given power allocations. Finally, by projecting R.H.S of  \eqref{power_opt_3} to the feasible region of the power allocations, we obtain \eqref{Power_update}. 
\end{proof}

Based on \textit{Proposition \ref{Power_Prop}}, an iterative algorithm to obtain transmit power allocations for a given set of transmitting UDs is provided as Algorithm 1. The convergence of Algorithm 1 is justified as follows.

\begin{algorithm}[t!]
	\caption{Transmit Power Allocations for A Given Set of Transmitting UDs} \label{alg:SLPA}
	\begin{algorithmic}[1]
		\STATE \textbf{Input:} Set of transmitting UDs, the file combinations, and the associated UDs with each transmitting UDs.
		
		\STATE \textbf{Initialize:} $\widehat{Q}_{u_k}=Q_o$, $\forall k=1,2,\cdots A $, $t=1$.
		
		\REPEAT
		
		\STATE Update the power allocation of the $u_k$-th transmitting UD, $\forall u_k $, by applying \eqref{Power_update}.
		
		\STATE Set $\widehat{Q}_{u_k}=Q_{u_k}$,  $\forall k=1,2,\cdots A$, and $t=t+1$
		
		\UNTIL{Objective function of \eqref{Power_opt} converges or $t > t_{\max}$.}
		
		\STATE \textbf{Output:} Final transmission power for all the transmitting UDs. 
		
	\end{algorithmic}
	
\end{algorithm}

\begin{proposition}
	Algorithm 1 provides  a stable and local optimal solution to $\eqref{Power_opt}$.
\end{proposition}

\begin{proof} We can proof \textit{Proposition 2} by resorting to the game theory. In fact, the proposed power allocation update can be considered as a non-cooperative power control game (NCPCG) where each transmitting UDs act as a rational and selfish player, and wants to maximize its utility by choosing the best possible power allocation strategy. To this end, the utility function of the $u_k$-th transmitting UD is given at the top of the next page, where $\mathbf{Q_{-k}}$ denotes the power allocation for the transmitting UDs other than the $u_k$-th UD. 
	
\begin{table*}
	\vspace*{-0.4cm}
	\begin{normalsize}	
\begin{equation*}
\label{Utility-K}
\begin{split}
\mathcal N_k(Q_k,\mathbf{Q_{-k}})=&\sum_{ u_i \in \mathtt u(\mathtt f_{u_k})} \log_2\left(1+ \cfrac{Q_{k} |\gamma_{u_k,u_i}|^{2}}{N_{0} +\sum_{m=1, m \neq k}^K Q_{m}| \gamma_{u_m,u_i}|^{2}}\right) \\
&+ \sum_{m=1, m\neq K}^K \sum_{ u_j \in \mathtt u(\mathtt f_{u_m})} \log_2\left(1+ \cfrac{Q_{m} |\gamma_{u_m,u_j}|^{2}}{N_{0} +\sum_{n=1, n \neq k,m}^K Q_{n}| \gamma_{u_n,u_i}|^{2}+Q_{k} |\gamma_{u_k,u_i}|^2}\right)
\end{split}
\end{equation*} 
\end{normalsize}
\vspace*{-0.5cm}
\hrulefill
\end{table*}

The utility function has two parts where the first part is the payoff in terms of the achievable throughput and the second term is the payoff for creating less interference to the other players in the system. Obviously, the first and second terms monotonically increase and decrease with the increase of transmission power, $Q_k$, respectively.  We denote the R.H.S of \eqref{Power_update} as $\mathcal{F}_k\left(\{\widehat{Q}_k\}\right)$. We can readily demonstrate that if $Q_k < \mathcal{F}_k\left(\{\widehat{Q}_k\}\right)$, $U_k(Q_k,\mathbf{Q_{-k}})$ monotonically increases, and if $Q_k> \mathcal{F}_k\left(\{\widehat{Q}_k\}\right)$, $U_k(Q_k,\mathbf{Q_{-k}})$ monotonically decreases. Therefore, $U_k(Q_k,\mathbf{Q_{-k}})$ is a quasi-concave utility function. From \cite[Theorem 3.2]{Game_Th}, for a non-cooperative game with quasi-concave utility functions, a Nash-equilibrium (NE) point must exists and it is obtained as the best response strategy of the players in the game. Note that, in an NE point, no player can improve its utility by taking an alternative strategy, and consequently, the overall solution must converge. We can easily justify that \eqref{Power_update} is same as the  best response strategy of the $k$-th transmitting UD, $\forall k$. Consequently, the iterative power allocation procedure, given in Algorithm 1, must converge to a stable point. We also emphasize that \eqref{Power_update} is derived by satisfying the Karush-khun-Tucker (KKT) conditions for \eqref{Power_opt}. Hence, a stable power allocation that is obtained by iteratively solving \eqref{Power_update} must converge to a local  optimal solution to \eqref{Power_opt}. Accordingly,  Algorithm 1 provides a stable and local optimal solution to \eqref{Power_opt}\footnote{The convergence of transmission power update equation, given by \eqref{Power_update}, is justified for asymptotically high signal-to-noise (SNR) ratio regime in \cite{Ahmed_Multi_level_EC}. However, using \textit{Proposition} 2, we justify that the considered power allocation converges without the assumption of asymptotic high SNR.}. 
\end{proof}

\subsection{Greedy Maximal Independent Set Selection Algorithm} \label{sec:gvCT}
In this sub-section, we describe a  maximal IS  selection  algorithm based on a greedy vertex search in the D2D-RA-IDNC graph and the priority of vertices defined in \pref{thm:Mweight1}.  Such a  greedy vertex search  approach was adopted   in \cite{9new,9newnew} without adopting the physical-layer rate, but  demonstrated its efficiency for completion time minimization. For simplicity, we use $v$ and $v'$ instead of $v^k_{r,i,h}$ and $v^k_{r,i',h'}$, respectively. Let $\mathcal E_{v,v'}$ be the adjacency connector of vertices $v$ and $v'$ in graph  $\mathcal{G}$ such that
\begin{equation}
\mathcal E_{v, v'} =
\begin{cases}
1 & \text{if $v$ is not adjacent to $v'$ in $\mathcal{G}$}, \\
0 & \text{otherwise}.
\end{cases}
\end{equation}

Further, let $g_v$ denote the weighted degree of vertex $v$, which can be expressed   as $g_v = \sum_{v'  \in \mathcal{G}} \mathcal E_{v, v'} \omega_o(v_v')$, where $\omega_o(v')$ is the priority  of vertex $v'$  defined in \eqref{eqw}.   Finally,  the modified weight  of vertex $v$ is defined as
\begin{align}\label{w:weight}
\omega_m (v) & = \omega_o(v) g_v  = 2^{N_w-d_{u_i}+1}\bar{\mathtt T}_{u_i}(t)\left(\frac{r}{B}\right) n_v.
\end{align}

To this end, at each step, the vertex search method adds a new vertex based on the maximum weight. Essentially, a vertex $v^*$ that has the maximum weight $\omega_m(v^*)$ is selected  and added to the maximal independent set $\mathtt I$ (i.e., $\mathtt I = \{v^*\}$). Then,  the  subgraph $\mathcal G(\mathtt I)$, which consists of vertices in graph  $\mathcal G$ that are not connected to vertex $v^*$ is extracted and considered for the next step. Next, a  new maximum weight  vertex $v^{'*}$ is selected from subgraph $\mathcal G(\mathtt I)$. We repeat this process until no more vertices that are not connected  to all the vertices in the maximal independent set  $\mathtt I$. The steps of the greedy vertex search selection are summarized in  \algref{alg:LGS}.

The  transmitting UD in $\mathtt I$  generates a coded  file by XORing all the files identified by the vertices in  $\mathtt I$.  It also adopts  the  transmission rate corresponding to  the vertices of  $\mathtt I$. It is worth mentioning that the MWIS  $\mathtt I$ and its corresponding modified weights in \eref{w:weight} provide the following potential benefits:
\begin{itemize}
	\item The modified weight of each vertex in  $\mathtt I$ shows the following. The first term $\left(\frac{r}{B}\right)$  provides a balance between  the transmission rate/power and the number of scheduled  UDs to transmitting UD $u_k$. The second term $2^{N_w-d_{u_i}+1}\bar{\mathtt T}_{u_i}(t)$ classifies the UDs based on their  completion time lower bounds. As such, we give them priority for scheduling. More importantly, through the weighted degree $g$, the modified weight of a vertex $v$ has a large original weight and it is not connected to a large number of vertices that have high original	weights. 
	\item Each UD is scheduled only to a transmitting UD that cached one of its missed files.
	\item The transmitting UD delivers an IDNC file with an adopted
	transmission rate/power that provides a lower completion time to a set of UDs. This adopted rate ensures the QoS rate guarantee and no larger than the channel capacities of all scheduled UDs.
\end{itemize} 

\subsection{Cross-layer NC Solution}
The iterative proposed CLNC solution maximizes the weighted sum rate subject to completion time reduction constraints, i.e., the problem of determining transmitting UDs and their transmission rates/powers and transmitted file combinations in a coordinated fashion. Particularly, we iterate between solving the
completion time reduction  problem for fixed transmit power and optimizing the
power level for a given schedule of completion time reduction. The main philosophy of this heuristic is to iteratively include more transmitting UDs and allocate transmission powers subject to  the reduction in the  completion time. At each iteration, it first determines the scheduled UDs by the set of chosen transmitting UDs as described in  Algorithm \ref{alg:LGS}. Then, given the resulting network-coded user scheduling, it executes a power allocation algorithm to determine  the power level of the transmitting UDs that maximizes the sum-rate and minimizes the completion time as described in Algorithm \ref{alg:SLPA}. The steps of the proposed iterative solution is described as in Algorithm \ref{alg:CLNC}. 

\begin{algorithm}[t!]
	\begin{algorithmic}[1]
		\STATE Generate D2D-RA-IDNC graph $\mathcal G$.
		\STATE Initialize $\mathtt I= \varnothing$.
		\STATE Set $\mathcal G(\mathtt I) \leftarrow \mathcal G$.
		\WHILE{$\mathcal G(\mathtt I) \neq \varnothing$}
		\STATE $\forall  v \in \mathcal G(\mathtt I)$: compute $\omega_o(v)$ and $\omega_m(v)$ using \eref{eqw} and \eref{w:weight}, respectively.
		\STATE Select $v^*=\arg\max_{v\in \mathcal G(\mathtt I)} \{\omega_m (v^*)\}$.
		\STATE Set $\mathtt I \leftarrow \mathtt I \cup v^*$.
		\STATE Obtain $\mathcal G(\mathtt I)$.
		\ENDWHILE
	\end{algorithmic}
	\caption{Greedy Maximal Weight Independent Set (MWIS) Selection} \label{alg:LGS}
\end{algorithm}

In Algorithm \ref{alg:CLNC}, $\mathcal{A}$ is the set of selected transmitting UDs, $\mathcal{M}_w$ is the set of UDs having non-empty \textit{Wants} set, and $\mathcal{X}$ is the set of all the targeted UDs.

\begin{algorithm}[t!]
\caption{Cross-layer Network-Coded (CLNC) Resource Scheduling Algorithm} \label{alg:CLNC}
	\begin{algorithmic}[1]
		\STATE \textbf{Initialize:} $\mathcal A$ = $\varnothing$, $\mathcal{M}_w=\mathcal{N}_{w}$, and $\mathcal{X}=\varnothing$.
		
		\STATE \textbf{Initialize:} Transmission power level $Q_o \in Q_\text{feasible}$ for each potential transmitting UD.
		
			\STATE Compute SINR(s) setting transmission power $Q_o$ and considering no interference.

		\REPEAT

		\STATE Construct the D2D-RA-IDNC graph using \sref{Graph} by considering using $\mathcal M_w$ as the set of potential transmitting or targeted 		UDs. 
		
		\STATE Select maximal independent set $\mathtt I$ using \algref{alg:LGS}.  Let the transmitting UD be $u_k$, the file combination  be $\mathtt f$ in the maximal independent set $\mathtt I$, and  the set of potential targeted UDs by $u_k$ be $\mathtt u(\mathtt f_{u_k})$. 
		
		\STATE Compute the lower bound on the individual completion time  of targeted UDs in $\mathcal X$ set using \eqref{eq:Aig} and increase delay of the non-targeted UDs in $\mathcal M_w \backslash \mathcal X$ by $\frac{B}{r(t)}$.
		
		\STATE \ignore{Compute the lower bound on the overall completion time $\bar{\mathtt T}$ and} Set $\mathcal A = \{\mathcal A, u_k\}$.
		
		\STATE Consider the UDs in $\mathcal M_w \backslash ( \{u_k\} \cup \mathcal X(\mathtt f_{u_k}))$ as future transmitting UDs or more targeted UDs.
		
		\STATE Compute the SINRs by setting the transmission power as $Q_o$ and considering interference from the UDs in the set $\mathcal A$.
		
		\ignore{\STATE Using the computed SINRs in step 10, construct a new reduced graph using the UDs in $\mathcal M_w \backslash ( \{u_k\} \cup \mathcal X(\mathtt f_{u_k}))$.
		
		\STATE Follow steps $6-8$ to determine a new transmitting UD $u_m$ and the targeted UDs in $\mathcal X(u_m)$.}
		
		\STATE Optimize the transmission power of the transmitting UDs $\mathcal A$ using Algorithm  \ref{alg:SLPA} to maximize the sum-capacity.
		
		\STATE If the sum-capacity is improved, $\mathcal{A} \leftarrow \mathcal{A} \cup u_i$ and $\mathcal{M}_w \leftarrow \mathcal{M}_w \setminus u_i$.

		\IF{$|\mathcal{A}| > 1$}

	\STATE For each receiving UD, $u_i$, compute $R_{u_i}$. If $R_{u_i} \geq R_{th}$ and $u_i \notin \mathcal{X}$,  update $\mathcal{X} \leftarrow \mathcal{X} \cup u_i$ and $\mathcal{M}_w \leftarrow \mathcal{M}_w \setminus u_i$. On the other hand, if  $R_{u_i} < R_{th}$ and $u_i \in \mathcal{X}$, update $\mathcal{X} \leftarrow \mathcal{X} \setminus u_i$ and $\mathcal{M}_w \leftarrow \mathcal{M}_w \cup u_i$. Repeat this step $\forall u_i \in \mathtt u(\mathtt f_{u_k})$ and $\forall u_k \in \mathcal A$.
	    
	    	\ENDIF	
	
		\STATE $\exists u_k \in \mathcal{A}$ such that the none of UDs in $\mathtt u(\mathtt f_{u_k})$ set satisfies the rate constraint,   $\mathcal{M}_w \leftarrow \mathcal{M}_w \cup u_k$ and $\mathcal{A} \leftarrow \mathcal{A} \setminus u_k$.
		
		\STATE Recompute $\bar{\mathtt T}$ and store the solution that achieves
		the minimum completion time.

		\UNTIL{No UDs can be added to the set $\mathcal{A}$.}
		\STATE \textbf{Output:} Overall completion time $\bar{\mathtt T}$.
		\end{algorithmic}
\end{algorithm}

\begin{proposition}
The CLNC solution achieves improved sum-rate compared to the interference free solution of \cite{15}.
\end{proposition}

\begin{proof} At each iteration of the proposed Algorithm 3, the set of transmitting UDs is updated. Let denote $\mathcal{A}^{(i)}$ be the set of transmitting UDs at the $i$-th iteration of Algorithm 3. Recall, only a finite number of UDs can be the transmitting UDs in a given TS. Thus, the set of  transmitting UDs is evolved  as
\begin{equation}
\label{evolution}
\mathcal{A}^{(1)} \rightarrow \mathcal{A}^{(2)}\rightarrow \cdots \rightarrow \mathcal{A}^{(\text{final})}. 
\end{equation}
Note that $\mathcal{A}^{(1)}$ contains only one transmitting UD. Particularly,  the proposed scheme initially selects a transmitting UD with maximum number of potential receiving UDs, and such a transmitting UD is included in $\mathcal{A}^{(1)}$. Subsequently, at each iteration Algorithm 3 adds one more transmitting UDs with the existing set of transmitting UDs given that the total sum-rate is improved. To this end, at each iteration, Algorithm 3 updates the power allocations of all the transmitting UDs to maximize the overall sum-capacity. Essentially, at each iteration of Algorithm 3, the sum-rate is non-decreasingly improved. Obviously, the sum-rate of the $\mathcal{A}^{(\text{final})}$ set must outperform the sum-rate achieved by the $\mathcal{A}^{(1)}$ set. Recall that the interference free solution of  \cite{15} selects only one transmitting UD each TS. Thus, the achievable sum-rate of the interference free solution of  \cite{15}  becomes same to the sum-rate achieved by the $\mathcal{A}^{(1)}$ set. Hence, the CLNC solution achieves improved sum-rate compared to the interference free solution of \cite{15}. 
\end{proof}
\textit{Remark 1:} \textit{By exploiting power allocation and time-varying channel of the UDs, the proposed CLNC solution activates multiple transmitting UDs at each TS. However, for severely strong inter-device interference channel, the power allocation may not improve the sum-capacity. In this case,  the proposed solution activates only one transmitting UD. Consequently, the the interference free solution of \cite{15} is a special case of the proposed CLNC solution. Hence, the proposed CLNC solution  always achieves a lower completion time  compared to  the  interference free solution.}

\textit{Remark 2:} \textit{When we schedule many UDs to the transmitting UDs, the number of targeted UDs is increased from the side information optimization, however, the sum-capacity may not be maximized. Consequently, we optimize the power/rate of the transmitting UDs to maximize the	sum-capacity. Thus, our CLNC solution not only increases the number of targeted UDs, but also maximizes the sum-capacity. }

\subsection{Complexity Analysis}
For any arbitrary D2D network setting and at any iteration of the proposed algorithm, we need to construct the D2D-RA-IDNC graph, calculate the power allocation of the transmitting UDs, and find the MWIS.

Since each UD caches only a set of files, the total number of vertices in D2D-RA-IDNC graph corresponding to that UD is $V=|\mathcal C|\times N$. Therefore, we construct the D2D-RA-IDNC graph for all UDs by generating $O(VN)$ vertices. Building the adjacency matrix needs a computational complexity of $O(V^2N^2)$. For the vertex search algorithm, we need first to calculate the weights of all vertices, and then finding the MWIS. It is easily to note that all UDs having vertices in the independent set have the same transmission rate as they initially corresponding to the same transmitting UD. Thus, the algorithm needs $|\mathcal{R}|$ maximal ISs. Note that
each maximal IS has at most $V$ vertices as
each UD can be targeted by at most one file (i.e., one
vertex for each targeted UD) per transmission. Each iteration with a given rate needs a complexity of $O(VN)$ for weight calculations of the MWIS. It also needs searching for at most $N-1$ vertices. Then, the complexity of the algorithm for finding the maximal ISs for all rates and their sum weights, at most, is $O(VN|\mathcal{R}|+(N-1)|\mathcal{R}|)=O(|\mathcal{R}|(VN+N-1))$. The computational complexity of constructing the D2D-RA-IDNC graph, building the adjacent matrix,  and  finding the MWIS is $O(|\mathcal{R}|(VN+N-1))+O(V^2N^2)=O(V^2N^2)$.

On the other hand, calculating the power allocation for any fixed D2D schedule needs $C_p=O(|\mathtt u_{u_1}|\times |\mathtt u_{u_2}| \times \cdots |\mathtt u_{u_K}|)$. Finally, Algorithm \ref{alg:CLNC} iterates between constructing the D2D-RA-IDNC graph and finding its corresponding MWIS and optimizing power levels of the transmitting UDs, thus leading to an overall computational complexity of $O(T(V^2N^2+C_p))$, where $T$ is the number of iterations.

\section{Numerical Results} \label{NC}
In this section, we present some numerical results that compare
the completion time performance of our proposed CLNC scheme with existing coded and uncoded schemes. We consider a D2D network where UDs are distributed  randomly within a  hexagonal cell of radius $500$m. We assume the channel gains between UDs follow the standard path-loss model, which consists of three components: 1) path-loss of $148 + 37.6 \log_{10}(d_{{u_k},{u_i}})$dB, where $d_{{u_k},{u_i}}$ represents the distance between $u_k$-th UD and $u_i$-th UD in km; 2) log-normal shadowing with $4$dB standard deviation and 3) Rayleigh channel fading with zero-mean and unit variance. We consider that the channels are perfectly estimated. The noise power and maximum’ UD
power are assumed to be $-174$ dBm/Hz and $Q_\text{max}=-42.60$ dBm/Hz, respectively, and the bandwidth is $1$ MHz. Unless otherwise stated, we initially consider that  each UD already has about $45\%$ and $55\%$ of $\mathcal F$ files for the considered schemes.
To evaluate the performance of our proposed scheme with different thresholds ($R_{\text{th}1}=0.5$, and $R_{\text{th}2}=5$), we simulate various scenarios with different
number of UDs, number of files, file sizes, and demand ratio of UDs. These thresholds represent the minimum transmission rates required for QoS. The performances of our joint solution for $R_{\text{th}1}=0.5$ and $R_{\text{th}2}=5$ are shown in solid and dash red lines, respectively.

For the sake of comparison, our proposed schemes are
compared with the following  existing schemes.
\begin{itemize}
	\item  \textbf{Uncoded Broadcast:} This scheme picks a random UD that broadcasts an uncoded file from its \textit{cache} set that is missing at the largest number of other UDs. Moreover, this scheme uses the minimum channel capacity from the transmitting UD to all other UDs as the transmission rate. 
	\item \textbf{Cooperative RLNC}: This RLNC algorithm  picks the UD with the highest side information rank  as the transmitting UD in a D2D  transmission \cite{19}. The picked   UD  encodes all  files  using random coefficient from a large Galois field. However, this algorithm discards  the dynamic transmission rates and for  the  transmission  to be successfully received by all other UDs, the minimum channel  capacity  from  the transmitting UD to all other UDs is adopted   as the  transmission rate.
	\item \textbf{Cooperative IDNC}: This  IDNC  algorithm considers cooperation among UDs and jointly  selects a set of transmitting UDs and their XOR file combinations \cite{9a}. However, this algorithm focuses on serving a large number of UDs with a new file in each time index to reduce  the overall completion time. Due to ignoring  the dynamic  rate adaptation, the minimum channel  capacity  from the transmitting UDs to all  targeted UDs is adopted as the transmission rate. 
\end{itemize}

For completeness of this work, we also compare our proposed scheme with the recent RA-IDNC work in \cite{15}. In this scheme, RA-IDNC scheme is employed for D2D network that allows only one transmitting UD to transmit at a time.

\begin{figure}[t]
	\centering
	\includegraphics[width=0.9\linewidth]{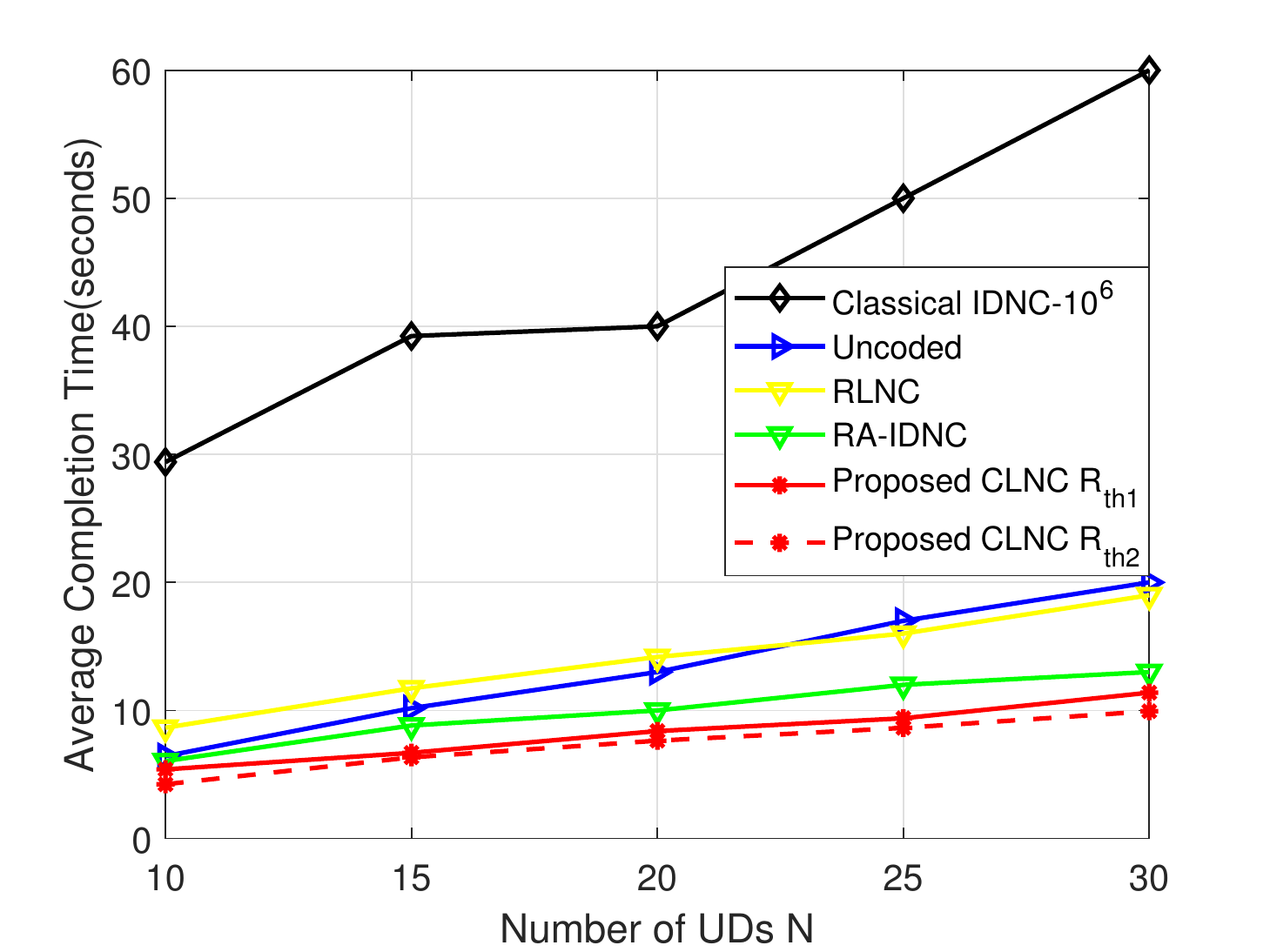}
	\caption{Average completion time in sec. vs the number of UDs $N$.}
	\label{fig2}
\end{figure}

\begin{figure}[t]
	\centering
	\includegraphics[width=0.9\linewidth]{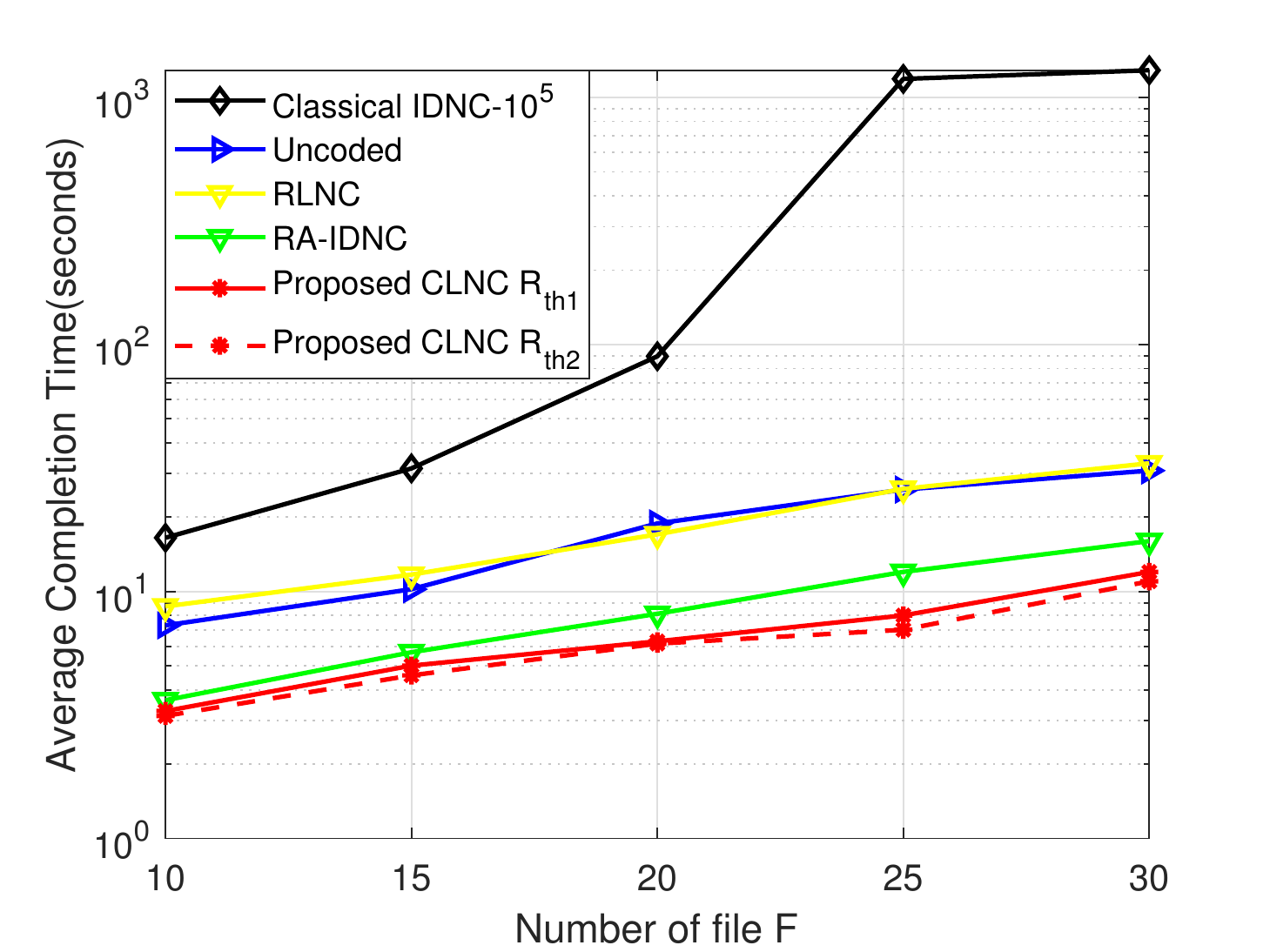}
	\caption{Average completion time in sec. vs the number of files $F$.}
	\label{fig3}
\end{figure}

In Fig. \ref{fig2},  we depict the  average completion time versus the number of UDs $N$. We consider a D2D model with a frame of $20$ files and a file size of $1$ Mbits. From this figure, we can observe 
that the proposed CLNC scheme offers an improved performance in terms of completion time minimization as compared to the other schemes for all considered number of UDs. This improved performance is due the fact that our
proposed scheme judiciously selects potential UDs for transmitting coded files to a set of schedule UDs, adopts the transmission rate, and optimizes the transmission power of each transmitting UD. This in turn aides the file combination selection process.
The uncoded broadcast scheme sacrifices the
rate optimality by scheduling the maximum number of UDs. Although uncoded scheme needs a fewer number of transmissions, at least $F$ transmissions, it requires longer transmission durations for frame delivery completion. This leads to a high completion time.  On the other hand, the RA-IDNC scheme improves the selection of file process by adapting the transmission rate, but it suffers from activating only one transmitting UD at each transmission slot. This is a clear limitation of the RA-IDNC scheme as it does
not fully exploit the simultaneous transmissions from multiple UDs. The proposed CLNC scheme strikes a balance between the aforementioned 
aspects by jointly selecting the number of targeted UDs and
the transmission rate of each transmitting UD such that the overall completion time is minimized. This results in a full utilization of simultaneous transmissions from multiple transmitting UDs. Consequently, an improved performance of our proposed scheme compared to the RA-IDNC scheme is achieved. Moreover, our proposed scheme improves the used rates using power control on each transmitting UD.

In Fig. \ref{fig3}, we show the average completion time versus the number of files $F$. Fig. \ref{fig3} considers different sizes of frames. The simulated D2D system composed of $20$ UDs and file size of $1$ Mbits. For the same reason as mentioned for Fig. \ref{fig2}, our proposed scheme outperforms other schemes. It can be observed from the figure that increasing the frame size leads to an increased completion time of all schemes. This is because for few  files, the opportunities of mixing files using IDNC in the proposed and other NC schemes are limited. As a result, all NC schemes have roughly similar performances. As the number of files increases, the  increase in the completion time with our proposed scheme is low.  This is due to the fact that our proposed scheme judiciously allows each transmitting UD to decide on a set of files to be XORed. As such, they are beneficial to a significant set of UDs that have relatively good channel qualities. Note that uncoded broadcast and RLNC schemes complete frame transmissions in fewer transmissions ($F$ transmissions) than our developed scheme. However, each of their transmission durations is longer than a single transmission of the proposed schemes since they are adopting the transmission rates to the minimum of all achievable capacities. 

\begin{figure}[t]
	\centering
	\includegraphics[width=0.9\linewidth]{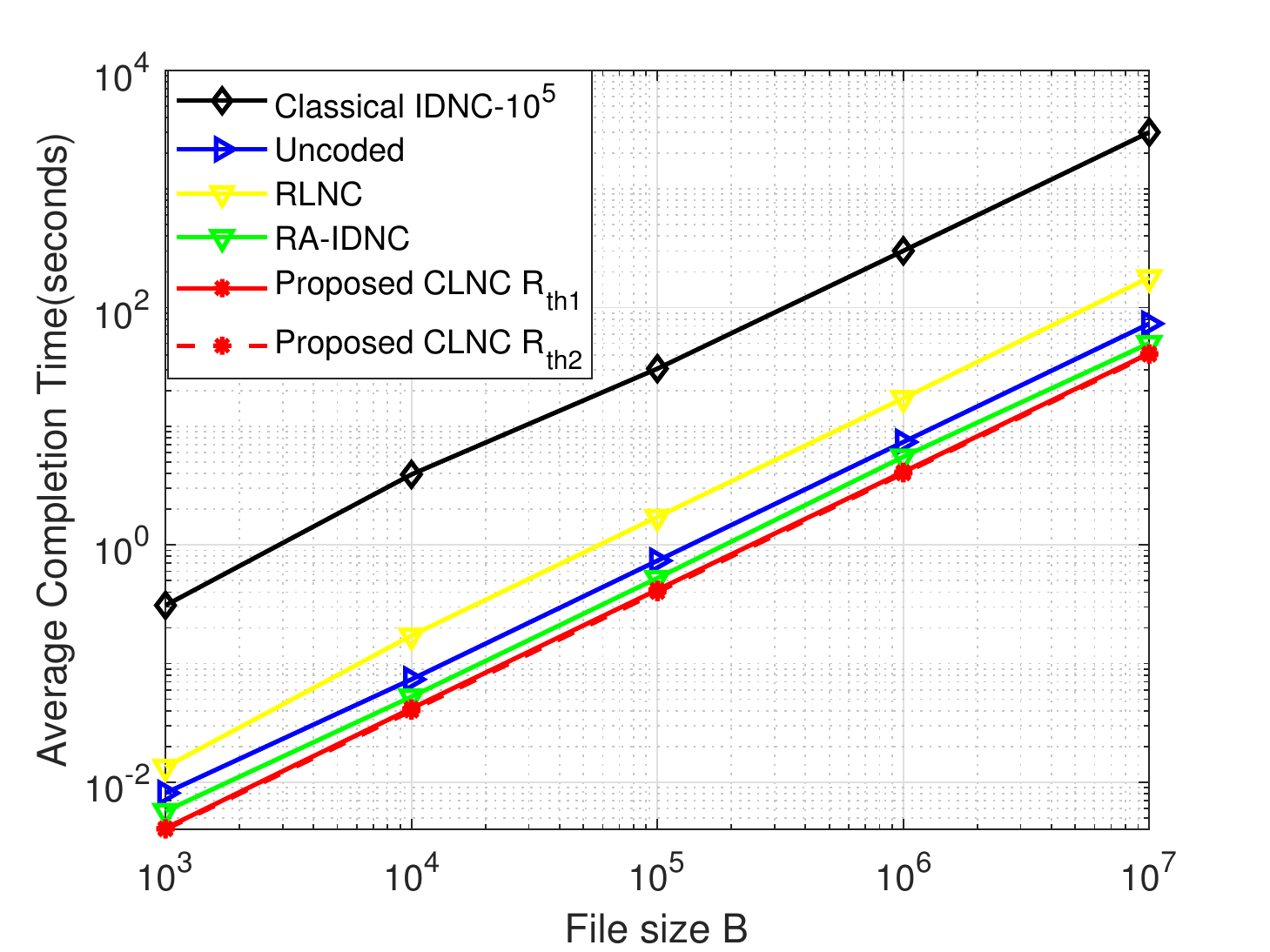}
	\caption{Average completion time in sec. vs file size $B$.}
	\label{fig4}
\end{figure}

\begin{figure}[t]
	\centering
	\includegraphics[width=0.9\linewidth]{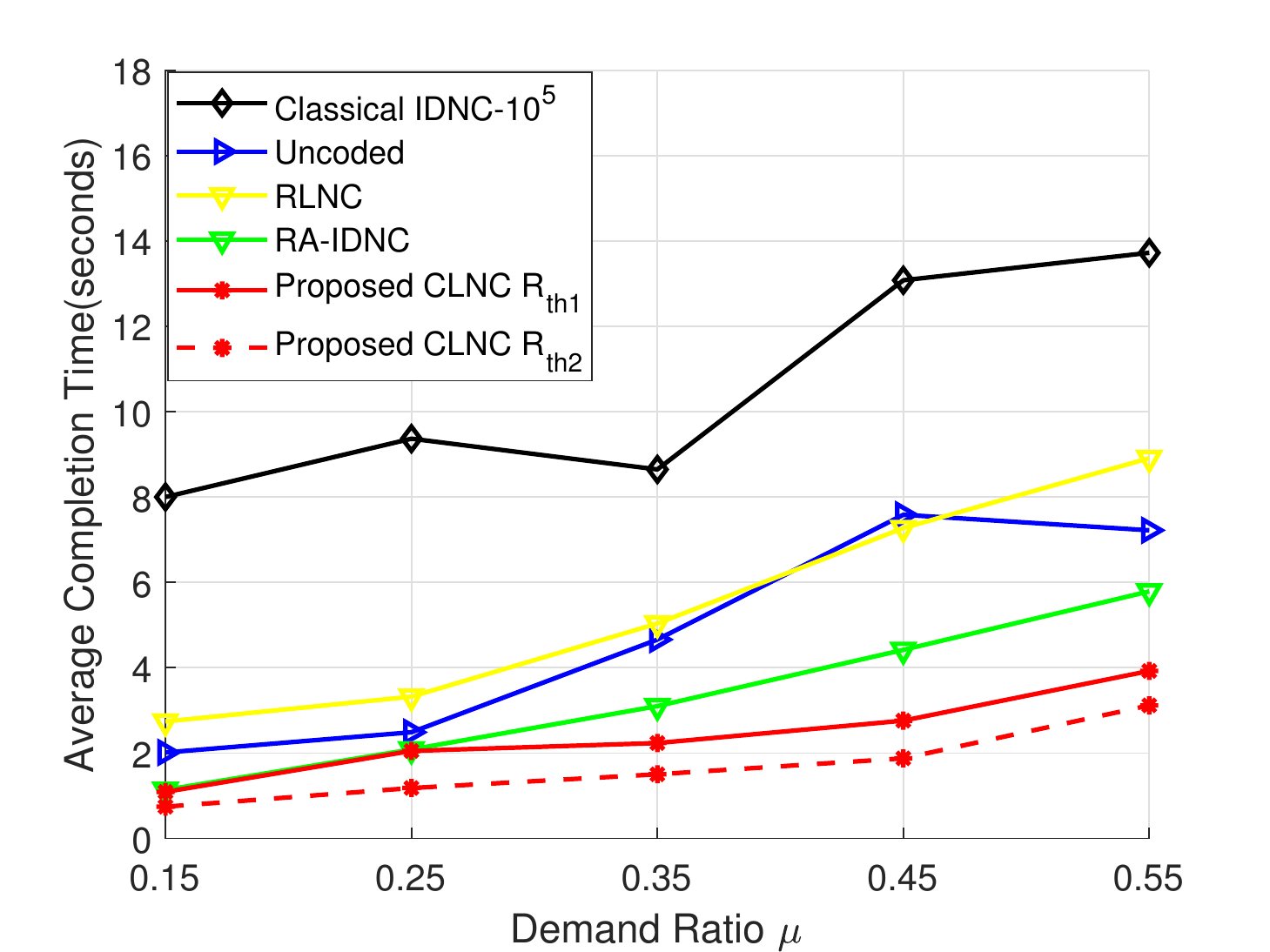}
	\caption{Average completion time in sec. vs demand ratio $\mu$.}
	\label{fig5}
\end{figure}

In Fig. \ref{fig4}, we illustrate the impact of increasing the file size $B$ on the average completion time. Fig. 8 shows the size of such popular files and  how long it takes for the proposed solution to deliver a complete frame to UDs. In this figure, we simulate the D2D system composed of $20$ UDs and $15$ files. We observe that the completion time performances of all schemes increase linearly with the file size. This agrees
with the completion time expression in \corref{cor:LBcompletion}, where it was shown that  $\mathtt T$ increases linearly with $B$. From physical-layer consideration, as $B$ increases, more bits are needed for delivering files. Thus, time delay is increased to receive files from transmitting UDs. It can be seen that the proposed scheme in all above figures outperforms all other schemes for different rate thresholds as shown with red lines. As the rate threshold increases, the completion time improvement increases. This is because as the rate threshold increases, a certain number of UDs is scheduled and the transmission rate of the transmitting UDs becomes high. Thus, the role of our proposed scheme for completion time minimization and QoS optimization technique becomes more noticeable.

In Fig. \ref{fig5}, we illustrate the impact of changing the demand ratio $\mu$ on the average completion time. This ratio represents the demand portion of the requested files of UDs. In this figure, we simulate the D2D system composed of $10$ UDs and $8$ files each with a size of $1$ Mbits. We can observe that the completion time performance of our proposed scheme outperforms the performances of other schemes for the whole range of $\mu$. It can be seen from the figure that increasing the demand ratio leads to an increased completion time of all schemes. This is because for high  demand ratio, the number of transmissions for delivering all the files to all the UDs of all schemes increases. As a result, the completion time performance of the considered schemes increases.

Finally, we provide some observations from our presented simulation results as follows. First, it is always beneficial from network-layer perspective to schedule many UDs with IDNC files as in the classical IDNC scheme. However, selecting the minimum transmission rate of all the scheduled UDs degrades the completion time performance of the classical scheme. Second, although the uncoded brodcast and RLNC schemes schedule almost all the UDs, they adopt the transmission rates to the minimum of all the scheduled UDs. Thus, its completion time performance is degraded, especially for large network sizes since selecting the minimum transmission rate of an increasing set is always minimum.   Third, RA-IDNC scheme overcomes the limitations of the aforementioned schemes but suffers from selecting only one transmitting UD. This limitation further degrades the completion time performance of the RA-IDNC scheme in large network sizes. This due to the fact that RA-IDNC scheme always selects one transmitter regardless of the size of the network. Conversely, our transmission framework is more practically  relevant as it considers different transmitting UDs and optimizes the employed rates using power control on each transmitting UD. 

\section{Conclusion} \label{C}
In this paper, we have studied the joint optimization of CLNC and D2D communications for the file delivery phase with the goal of minimizing the completion time  while guaranteeing UD’s QoS, subject to
the UD’s cache files, the required minimum rate, power
allocation, and NC constraints. The completion time minimization problem in interference-allowed setup is solved over a set of transmitting UDs, their power allocation, dynamic rate selection and transmitted file combinations. By using a graph theory technique, we proposed a novel and efficient
approach that uses cross-layer NC for power optimization and UDs coordinated scheduling in D2D networks. Specifically, our proposed solution judiciously iterates  between finding the MWIS in  the D2D-RA-IDNC graph and optimizing the power allocation, subject to the resultant interference of the newly added transmitting UDs. Simulation results show that the proposed interference-allowed solution  reduces the  completion time  compared to the  interference-free solution as well as  conventional  network coding algorithms.

\end{document}